\documentclass[11pt,a4paper]{article}

\usepackage[a4paper,text={150mm,240mm},centering,headsep=10mm,footskip=15mm]{geometry}
\usepackage[T1]{fontenc}
\usepackage{epsfig, float,color}
\usepackage{amsmath,amssymb, amscd}
\usepackage{amsfonts,amsbsy}
\usepackage{hyperref}
\usepackage{url}
\usepackage{bbm}
\usepackage{mathrsfs} 
\usepackage{cite}

\newcommand{\R}{\mathbb{R}}
\newcommand{\C}{\mathbb{C}}

\newcommand{\id}{\mathbbm{1}}
\newcommand{\vx}{{\mathbf{x}}}

\newcommand{\Banach}{\mathscr{B}}

\newcommand{\ret}{{\rm ret}}
\newcommand{\adv}{{\rm adv}}
\newcommand{\sym}{{\rm sym}}
\newcommand{\free}{{\rm free}}
\newcommand{\be}{\begin{equation}}
\newcommand{\ee}{\end{equation}}

\DeclareMathOperator*{\esssup}{ess \, sup}
\DeclareMathOperator{\sgn}{sgn}

\newtheorem{theorem}{Theorem}[section]

\newenvironment{proof}[1][Proof:]{\begin{trivlist}
\item[\hskip \labelsep {\bfseries #1}]}{\end{trivlist}}

\newenvironment{remark}[1][Remark:]{\begin{trivlist}
\item[\hskip \labelsep {\bfseries #1}]}{\end{trivlist}}

\newcommand{\qed}{\hfill\ensuremath{\square}}

\title{Interacting relativistic quantum dynamics of two particles\\ on spacetimes with a Big Bang singularity}

\author{
Matthias Lienert\thanks{Fachbereich Mathematik, Eberhard-Karls-Universit\"at, Auf der Morgenstelle 10, 72076 T\"ubingen, Germany.
     E-mail: lienertmat@gmail.com}\ \ and
Roderich Tumulka\footnote{Fachbereich Mathematik, Eberhard-Karls-Universit\"at, Auf der Morgenstelle 10, 72076 T\"ubingen, Germany. E-mail: roderich.tumulka@uni-tuebingen.de}
}

\date{April 16, 2019}

\begin{document}

\maketitle

  \begin{abstract}
\noindent Relativistic quantum theories are usually thought of as being quantum field theories, but this is not the only possibility. Here we consider relativistic quantum theories with a fixed number of particles that interact neither through potentials nor through exchange of bosons. Instead, the interaction can occur directly along light cones, in a way similar to the Wheeler-Feynman formulation of classical electrodynamics. For two particles, the wave function is here of the form $\psi(x_1,x_2)$, where $x_1$ and $x_2$ are spacetime points. Specifically, we consider a natural class of covariant equations governing the time evolution of $\psi$ involving integration over light cones, or even more general spacetime regions. It is not obvious, however, whether these equations possess a unique solution for every initial datum. We prove for Friedmann-Lema\^itre-Robertson-Walker spacetimes  that in the case of purely retarded interactions there does, in fact, exist a unique solution for every datum on the initial hypersurface. The proof is based on carrying over similar results for a Minkowski half-space (i.e., the future of a spacelike hyperplane) to curved spacetime. Furthermore, we show that also in the case of time-symmetric interactions and for spacetimes with both a Big Bang and a Big Crunch solutions do exist. However, initial data are then not appropriate anymore; the solution space gets parametrized in a different way.
	
  	\vspace{0.3cm}
  	
  \end{abstract}

%%%%%%%%%%%%
%%%%%%%%%%%%
\section{Introduction} \label{sec:intro}

In this paper, we study a certain type of time evolution law for a relativistic quantum-mechanical two-particle wave function $\psi(x_1,x_2)$ where $x_1, x_2$ are spacetime points. These wave functions, first suggested by Dirac in 1932 \cite{dirac_32} and used in a similar way by Tomonaga \cite{tomonaga} and Schwinger \cite{schwinger}, are called \textit{multi-time wave functions} because of the occurrence of many time coordinates in the argument of $\psi$. Their merit is to provide a manifestly covariant representation of a multi-particle quantum state. On Minkowski spacetime, they are moreover straightforwardly related to the usual non-relativistic wave functions of quantum mechanics, $\varphi(t,\vx_1,\vx_2)$ with $\vx_i\in \R^3$, by evaluation at equal times, i.e.,

\be
	\varphi(t,\vx_1,\vx_2)=\psi(t,\vx_1,t,\vx_2).
	\label{eq:relationpsiphi}
\ee
A recent overview of the theory of multi-time wave functions is provided in \cite{dice_paper}.

Interestingly, \textit{multi-time wave functions permit to formulate new types of interacting quantum dynamics that cannot so easily be expressed in the other pictures of quantum theory.} In particular \cite{direct_interaction_quantum}, multi-time wave functions make it possible to express direct relativistic interactions (not mediated by fields) at the quantum level, much as in the Wheeler-Feynman formulation of classical electrodynamics \cite{WF1,WF2}. The crucial point is that such interactions happen with a time delay, and contrary to the usual Schr\"odinger picture wave function $\varphi(t,\vx_1,\vx_2)$, a multi-time wave function $\psi(x_1,x_2)$ can be considered on configurations $(x_1,x_2)$ where $x_1,x_2$ are separated in time by this delay. (The most natural possibility is that interactions happen exactly at Minkowski distance-squared zero, i.e., along light cones, as in the Wheeler-Feynman theory.)

The following class of integral equations has been suggested as a way to define a dynamics for $\psi$ with direct relativistic interactions \cite{direct_interaction_quantum}.
\be
	\psi(x_1,x_2) = \psi^\free(x_1,x_2) + \lambda \int dV(x_1') \int dV(x_2') \, G_1(x_1-x_1') G_2(x_2-x_2') K(x_1',x_2') \psi(x_1',x_2').
	\label{eq:inteq}
\ee
We will call these equations \textit{multi-time integral equations}. Here, $G_1,G_2$ are Green's functions of relativistic wave equations, such as the Klein-Gordon (KG) or Dirac equations, which describe the free (non-interacting) dynamics of particles 1 and 2, respectively; $\psi^\free$ is a solution of these equations in the respective particle's variables, $\lambda$ a coupling constant, $dV(x_i)$ are the infinitesimal spacetime volume elements, and $K$ is the so-called \textit{interaction kernel}. On Minkowski spacetime, direct interactions along light cones are expressed by $K(x_1,x_2) \propto \delta((x_1-x_2)^2)$ where $(x_1-x_2)^2$ is the Minkowski distance-squared of $x_1$ and $x_2$. Other choices of $K$ are also possible; they correspond to different types of interactions. In any case, it is important that $K(x_1,x_2)$ be a covariant object in order for \eqref{eq:inteq} to be covariant as well.

Let us sketch briefly what makes the integral equation \eqref{eq:inteq} a natural possibility. In fact, the initial value problem $\varphi(0,\vx_1,0,\vx_2)=\varphi_0(\vx_1,\vx_2)$ for the non-relativistic two-particle Schr\"odinger equation
\be
	i \partial_t\varphi(t,\vx_1,\vx_2) = \left( H_1^\free + H_2^\free + \lambda V(t,\vx_1,\vx_2) \right) \varphi(t,\vx_1,\vx_2)
\ee
can equivalently be formulated as the following integral equation:
\begin{align}
	\varphi(t,\vx_1,\vx_2) = \varphi^\free(t,\vx_1,\vx_2) +& \lambda \int_0^\infty dt' \int d^3 \vx_1' \, d^3 \vx_2' ~G_1(t-t',\vx_1-\vx_1')\nonumber\\
& \times G_2(t-t',\vx_2-\vx_2') V(t',\vx_1',\vx_2') \varphi(t',\vx_1',\vx_2').
\label{eq:schroedint}
\end{align}
Here, $\varphi^\free(t,\vx_1,\vx_2) = (e^{-i(H_1^\free + H_2^\free)t} \varphi_0)(\vx_1,\vx_2)$ is the solution of the free Schr\"odinger equation (without the potential $V$) for the initial data $\varphi_0$, and $G_k$, $k=1,2$ are the retarded Green's functions of the operators $(i\partial_t-H_k^\free)$, respectively. Bearing in mind Lorentz covariance and the relation \eqref{eq:relationpsiphi}, one can now see why \eqref{eq:inteq} is a natural generalization of \eqref{eq:schroedint} to the relativistic case. Moreover, the comparison suggests that it is natural to consider retarded Green's functions $G_i^\ret$. Besides, we also consider the symmetric Green's functions $G^\sym_i = \frac{1}{2}(G^\ret_i + G^\adv_i)$ a natural choice, as it leads to time-reversal invariance (on Minkowski spacetime). Note that this choice of Green's functions $G_i$ determines whether the interaction term depends on the past and/or on the future.

For the theory of multi-time wave functions, integral equations of the form \eqref{eq:inteq} are interesting because they create a new mechanism for relativistic interactions. Relativistic interactions are notoriously difficult to achieve for multi-time wave functions \cite{nogo_potentials,deckert_nickel_2016}, and have so far only been realized for a few examples \cite{drozvincent_1981,2bd,qftmultitime,multitime_pair_creation,1d_model,nt_model}.

Another source of motivation for considering \eqref{eq:inteq} is that the well-known Bethe-Salpeter equation, which is derived from QFT and describes relativistic bound states of two particles (see \cite{bs_equation} and \cite[chap. 6]{greiner_qed}), has a similar form,\footnote{The differences are discussed in \cite{direct_interaction_quantum}.} with $G_1, G_2$ given by Feynman propagators and $K$ by an infinite series of Feynman diagrams.

As \eqref{eq:inteq} constitutes a new type of evolution equation, it is important to demonstrate that it is mathematically well-defined by proving the existence and uniqueness of solutions. First results in this direction have been obtained in a preceding paper by the authors \cite{mtve}. It has been shown for simplified models in flat spacetime that \eqref{eq:inteq} has a unique solution $\psi$ for every $\psi^\free$. These models focus on the Klein-Gordon case and use retarded Green's functions---so that the interaction term depends only on the past. Then arbitrary bounded interaction kernels and certain singular interaction kernels with a special type of singularity have been covered. Perhaps the most crucial assumption in \cite{mtve} was to postulate a beginning in time, i.e., such that the time integrals in \eqref{eq:inteq} begin not at $-\infty$ but at 0. This greatly simplifies the problem, as \eqref{eq:inteq} then attains a Volterra structure in the time variables. That means that the time variables are integrated merely from 0 to $t_k$; i.e., the integral operator in \eqref{eq:inteq} is of the form
\be\label{Volterra}
\int_0^{t_1} dt_1' \int d^3\vx_1' \int_0^{t_2} dt_2' \int d^3\vx_2' \: \kappa(x_1,x_2,x_1',x_2') \, \psi(x_1',x_2')
\ee
with $x_k=(t_k,\vx_k)$ etc. The beginning in time at $t=0$ was a crude model for
the Big Bang singularity which our universe is believed to have. Of course, the Big Bang singularity should be implemented in a physically natural way, not just by cutting off the time axis at $t=0$. Hence the need to find an appropriate generalization of the integral equation \eqref{eq:inteq} for curved spacetimes which naturally feature a Big Bang singularity, such as Friedmann-Lema\^itre-Robertson-Walker (FLRW) spacetimes. Then one has to show that the Big Bang singularity is still compatible with the existence and uniqueness of solutions. To address this circle of questions is the goal of the present paper.

We shall approach the problem as follows. In Sec. \ref{sec:previousresults}, we review the explicit form of the integral equation \eqref{eq:inteq} on a Minkowski half-space with 1+$d$ dimensions for $d=1,2,3$ and summarize the results of the previous paper \cite{mtve}. (The details will be needed later.) Next, we show how to formulate, in general terms, the integral equation on arbitrary curved spacetimes (Sec. \ref{sec:inteqcurvedgeneral}). Sec. \ref{sec:flrw} briefly summarizes the main facts about FLRW spacetimes. Sec. \ref{sec:greensfnscurved} then deals with finding an explicit formulation of the integral equation for particular cases. This mainly requires calculating the Green's functions $G_1, G_2$, which is done for the case of massless scalar (Klein-Gordon) particles on various FLRW spacetimes using conformal invariance. We arrive at a number of explicit models for curved spacetimes (Sec. \ref{sec:explicitexamples}). Sec. \ref{sec:existence} contains our main results. In Sec. \ref{sec:existenceflat}, we show that the theorems in \cite{mtve} carry over to the case of flat FLRW universes. The main result are rigorous existence and uniqueness theorems for a class of integral equations with sufficiently regular interaction kernels for $d=1,2,3$ dimensions (Thms. \ref{thm:flatflrw} and \ref{thm:singkernelflatflrw}). Remarkably, these results cover a class of manifestly covariant interacting models in 1+3 spacetime dimensions. In Sec. \ref{sec:existenceclosed}, we study the case of time-symmetric interactions for closed FLRW spacetimes. The main result here is an existence and uniqueness theorem for small coupling constants $\lambda$ (Thm. \ref{thm:existenceclosedsingular}). In Sec. \ref{sec:conclusion}, we conclude and point out interesting problems for future research.

%%%%%%%%%%%%
%%%%%%%%%%%%
\section{Previous results for a Minkowski half-space} \label{sec:previousresults}

In this section, we formulate the integral equation \eqref{eq:inteq} on a (1+$d$)--dimensional Minkowski half-space for $d=1,2,3$. After that, we state the existence and uniqueness results of the previous paper \cite{mtve}. We shall give the full details, as they are necessary later in the paper.

Minkowski half-space here means $[0,\infty) \times \R^d$ equipped with the Minkowski metric. We shall denote this spacetime by $\frac{1}{2} \R^{1,d}$.
For the rest of the paper, we furthermore set $\hbar = c =1$ and use the metric signature $+\underbrace{- \cdots -}_d$. A point $x \in \frac{1}{2} \R^{1,d}$ is denoted by $x=(t,\vx)$ with $\vx \in \R^d$. Moreover, we use the abbreviation $x^2 = t^2 -|\vx|^2$. 

Eq. \eqref{eq:inteq} becomes fully specified by the choices (a) of the free wave equations for particles 1 and 2, (b) of the particular Green's functions of these equations, and (c) of the interaction kernel $K$. Here, we shall choose the Klein-Gordon (KG) equations
\be
	(\square_k + m_k) \psi^\free = 0,~~k=1,2
\ee
as the free wave equations. Moreover, we focus on the case of retarded Green's functions, $G_k = G_k^\ret$. As detailed in \cite{direct_interaction_quantum}, a very natural choice of the interaction kernel is $K(x_1,x_2) = G^\sym$, where $G^\sym$ stands for the symmetric Green's function of the massless KG equation (i.e., the wave equation) in the respective spacetime dimension $d$. The reasons for this are (i) interactions along light cones are given by $K(x_1,x_2) = \frac{1}{2\pi} \delta((x_1-x_2)^2) = G^\sym(x_1,x_2)$ on Minkowski spacetime, and (ii) for indistinguishable particles $K(x_1,x_2) = K(x_2,x_1)$ is necessary (and satisfied by $K = G^\sym$). 
That being said, other choices for $K$ are also possible, corresponding to different kinds of interactions. Note also that the above arguments concern the choice of $K$, not of $G_1,G_2$ where the situation is different.

\paragraph{Green's functions.} To specify \eqref{eq:inteq}, we need to know the Green's functions of the KG equation. 
The symmetric Green's functions are given by (see \cite[chap. 7.4]{zauderer}, \cite[appendix E]{birula_qed}):
\begin{table}[h]
\centering
\begin{tabular}{|l|l|}
\hline
	&\\[-3mm]
	$d$ & $G^\sym$\\[2mm]
	\hline \hline 
	&\\[-3mm]
	1 & $\frac{1}{2} H(x^2) J_0(m \sqrt{x^2})$\\[2mm]
	\hline
	&\\[-3mm]
	2 & $\frac{1}{2\pi} H(x^2) \frac{\cos(m\sqrt{x^2})}{\sqrt{x^2}}$\\[2mm]
	\hline
	&\\[-3mm]
	3 & $\frac{1}{2\pi} \delta(x^2) - \frac{m}{4\pi \sqrt{x^2}} H(x^2) J_1(m\sqrt{x^2})$\\[2mm]
	\hline
\end{tabular}
\caption{Symmetric Green's functions of the KG equation on Minkowski spacetime}
\label{tab:greensfnsmink}
\end{table}

\noindent Here, $H(s)$ denotes the Heaviside function and $J_0, J_1$ are Bessel functions of the first kind.
From $G^\sym$, one obtains the retarded Green's functions by
\be
	G^\ret(x) = H(x^0) \, G^\sym.
\ee

These ingredients allow us to write down Eq. \eqref{eq:inteq} on $\frac{1}{2} \R^{1,d}$ for the various dimensions $d$. We shall do this for arbitrary interaction kernels. Note that the Green's functions on $\R^{1,d}$ and $\frac{1}{2}\R^{1,d}$ have the same form. 

%%%
\paragraph{d=1:}
\begin{align}
	&\psi(t_1,z_1,t_2,z_2) = \psi^\free(t_1,z_1,t_2,z_2) + \frac{\lambda}{4} \int_0^{t_1} dt_1' \int_0^{t_2} dt_2' \int dz_1' \, dz_2'~ H(t_1-t_1'-|z_1-z_1'|)\nonumber\\
&~\times~ J_0\Big(m_1\sqrt{(t_1-t_1')^2-|z_1-z_1'|^2}\Big) \, H(t_2-t_2'-|z_2-z_2'|) \, J_0\Big(m_2\sqrt{(t_2-t_2')^2-|z_2-z_2'|^2}\Big)\nonumber\\[1mm]
&~\times~ K(t_1',z_1',t_2',z_2') \, \psi(t_1',z_1',t_2',z_2').
\label{eq:inteq1dsimplified}
\end{align}
Note that in $d=1$ the natural interaction kernel $K(t_1,z_1,t_2,z_2) = \frac{1}{2} \, H( (t_1-t_2)^2-|z_1-z_2|^2)$ (corresponding to $m=0$) is bounded.

%%%
\paragraph{d=2:}
\begin{align}
	&\psi(t_1,\vx_1,t_2,\vx_2) = \psi^\free(t_1,\vx_1,t_2,\vx_2) + \frac{\lambda}{(2\pi)^2} \int_0^{t_1} dt_1' \int_0^{t_2} dt_2' \int d^2 \vx_1'\, d^2 \vx_2'~\nonumber\\
& ~\times~ H(t_1-t_1'-|\vx_1-\vx_1'|)\, \frac{\cos\big(m_1\sqrt{(t_1-t_1')^2-|\vx_1-\vx_1'|^2}\big)}{\sqrt{(t_1-t_1')^2-|\vx_1-\vx_1'|^2}} \,
H(t_2-t_2'-|\vx_2-\vx_2'|) \, \nonumber\\
&~\times~ \frac{\cos\big(m_2\sqrt{(t_2-t_2')^2-|\vx_2-\vx_2'|^2}\big)}{\sqrt{(t_2-t_2')^2-|\vx_2-\vx_2'|^2}}\, K(t_1',\vx_1',t_2',\vx_2') \,  \psi(t_1',\vx_1',t_2',\vx_2').
\label{eq:inteq2dsimplified}
\end{align}

%%%
\paragraph{d=3:} For simplicity, we consider the massless case here. (The most singular terms are still included in the equation.) We obtain the equation:
\begin{align}
	&\psi(t_1,\vx_1,t_2,\vx_2) = \psi^\free(t_1,\vx_1,t_2,\vx_2) + \frac{\lambda}{(4\pi)^2} \int_0^{t_1} dt_1' \int d^3 \vx_1' \int_0^{t_2} dt_2' \int d^3 \vx_2'~  \nonumber\\
& ~\times~\frac{\delta(t_1-t_1'-|\vx_1-\vx_1'|)}{|\vx_1-\vx_1'|}\, \frac{\delta(t_2-t_2'-|\vx_2-\vx_2'|)}{|\vx_2-\vx_2'|} \, 
K(t_1',\vx_1',t_2',\vx_2') \, \psi(t_1',\vx_1',t_2',\vx_2').
\label{eq:inteq3d}
\end{align}
Formally integrating out the $\delta$-functions leads to:
\begin{multline}
	\hspace{-3mm}\psi(t_1,\vx_1,t_2,\vx_2) = \psi^\free(t_1,\vx_1,t_2,\vx_2) + \frac{\lambda}{(4\pi)^2}  \int d^3 \vx_1' \, d^3 \vx_2'~  \frac{H(t_1-|\vx_1-\vx_1'|)}{|\vx_1-\vx_1'|}\, \frac{H(t_2-|\vx_2-\vx_2'|)}{|\vx_2-\vx_2'|} \\[2mm]
 \times~ K(t_1-|\vx_1-\vx_1'|,\vx_1',t_2-|\vx_2-\vx_2'|,\vx_2') \,  \psi(t_1-|\vx_1-\vx_1'|,\vx_1',t_2-|\vx_2-\vx_2'|,\vx_2'),
\label{eq:inteq3dsimplified}
\end{multline}
where the Heaviside functions result from the lower limits of $t_1,t_2$ in \eqref{eq:inteq3d}. We shall consider this equation directly instead of \eqref{eq:inteq3d}.

 With these preparations, we are ready to state the previous results from \cite{mtve}.

%%%
\paragraph{Results for bounded interaction kernels.}
\begin{theorem}[See {\cite[thms. 3.2, 3.3, 3.4]{mtve}}.]\label{thm:boundedkernels}
	Let $d\in\{1,2,3\}$ and $K: \R^{2(1+d)} \rightarrow \C$ be a bounded function. Let furthermore $T>0$ and consider the Banach space
\be
    \Banach_d := L^\infty \big([0,T]^2_{(t_1,t_2)}, L^2(\R^{2d}_{(\vx_1,\vx_2)}) \big),
    \label{eq:banach}
\ee
with norm
\be
    \| \psi\|_{\Banach_d} = \esssup_{t_1,t_2 \in [0,T]} \| \psi(t_1,\cdot,t_2,\cdot) \|_{L^2} \:.
    \label{eq:norm}
\ee
Then for every $\psi^\free \in \Banach_d$, the above-mentioned multi-time integral equation for that dimension $d$ (one of the equations \eqref{eq:inteq1dsimplified}, \eqref{eq:inteq2dsimplified}, \eqref{eq:inteq3dsimplified}) has a unique solution $\psi \in \Banach_d$.
\end{theorem}

%%%
\paragraph{Result for a special singular interaction kernel in $d=3$.}
Apart from bounded interaction kernels, also the following singular interaction kernel for $d=3$ has been studied in \cite{mtve}:
\be
	K(t_1,\vx_1,t_2,\vx_2) = \frac{f(t_1,\vx_1,t_2,\vx_2)}{|\vx_1-\vx_2|},
\label{eq:specialsingkernel}
\ee
where $f:\R^8 \rightarrow \C$ is an arbitrary bounded function. This kernel imitates the structure of the natural interaction kernel in $d=3$, $K(t_1,\vx_1,t_2,\vx_2) = \delta((t_1-t_2)^2-|\vx_1-\vx_2|^2) = \frac{1}{2 |\vx_1-\vx_2|} [\delta(t_1-t_2-|\vx_1-\vx_2|) + \delta(t_1-t_2+|\vx_1-\vx_2|)]$. The difference is that the two $\delta$-functions have been replaced by a bounded function $f$. We then have:

\begin{theorem}[See {\cite[thm. 3.5]{mtve}}.] \label{thm:singularkernel3d}
	For every bounded $f:\R^8\to \C$ and every $\psi^\free \in \Banach_3$, Eq. \eqref{eq:inteq3dsimplified} with interaction kernel \eqref{eq:specialsingkernel} possesses a unique solution $\psi \in \Banach_3$.
\end{theorem}

%%%
\paragraph{Remark.} The choice of the function space $\Banach_d = L^\infty \big([0,T]^2_{(t_1,t_2)}, L^2(\R^{2d}_{(\vx_1,\vx_2)}) \big)$ is motivated by the expectation that in quantum physics the spatial $L^2$ norm of $\psi$ should be finite for each fixed time.  Besides, as $\psi^\free$ appears in the integral equation, $\Banach_d$ has to be such that $\psi^\free$, the solution of the free KG equation, lies in $\Banach_d$, which is the case for our choice. 
Some readers have suggested replacing $\Banach_d$ by a space of functions $\psi$ on $\R^8$ for which $\|\psi(t_1,\cdot,t_2,\cdot)\|_{L^2}$ falls off as $t_1,t_2\rightarrow \pm \infty$, but such a space would not include $\psi^\free$. 

%%%%%%%%%%%%
%%%%%%%%%%%%
\section{The integral equation on curved spacetimes} \label{sec:inteqcurved}

In this section, we find a natural analog of the integral equation \eqref{eq:inteq} on curved spacetimes. Then we calculate the Green's functions of the massless KG equation on flat and closed FLRW spacetimes. This allows us to write down a number of explicit examples for the integral equation on these spacetimes.

%%%%%
\subsection{Formulation of the integral equation on general curved spacetimes} \label{sec:inteqcurvedgeneral}

In order to extend the integral equation \eqref{eq:inteq} to curved spacetimes, we  examine how to generalize its ingredients.
\begin{enumerate}
	\item[(a)] \textit{Free wave equations.} The most important relativistic quantum-mechanical wave equations, the Dirac equation and the KG equation, have a canonical generalization to curved spacetimes. We shall focus on the KG case for simplicity. Now, the KG equation with the appropriate coupling to the scalar curvature of a curved spacetime with Lorentzian metric $g$ reads \cite{john,penrose_1964}:
\be
	(\square_g + m^2 - \xi R)\psi = 0,
	\label{eq:curvedkg}
\ee
where $\square_g = \nabla_\mu \nabla^\mu$ and $\nabla_\mu$ denotes the covariant derivative with respect to the metric $g_{ab}$ on the spacetime manifold $\mathcal{M}$, as well as $\xi = (d-1)/(4d) = \frac{1}{6}$ for $d=3$ space dimensions; $R$ stands for the Ricci scalar. To include the term $\xi R\, \psi$ is natural because \eqref{eq:curvedkg} is, like the massless Dirac equation, invariant under conformal transformations, see \eqref{eq:kgconformaltrafo} below.
So, $\psi^\free$ in \eqref{eq:inteq} is taken to be a solution of 
\be\label{freeKGcurved}
(\square_{g,k} + m_k^2 - \xi R)\psi^\free(x_1,x_2) = 0,~~~k=1,2.
\ee
\item[(b)] \textit{Spacetime volume elements.} The spacetime volume elements $dV(x_i)$ for curved spacetime are given by $dV(x_i) = \sqrt{-g(x)} \, d^{1+d}x$ where $g(x)$ is the metric determinant and $d^{1+d}x$ stands for the coordinate volume element.
\item[(c)] \textit{Green's functions} of \eqref{eq:curvedkg} are bi-scalar distributions $G(x,x')$ which satisfy
 \be
	(\square_g + m^2 - \xi R(x)) G(x,x') = [-g(x)]^{-1/2} \delta^{(1+d)}(x,x'),
	\label{eq:kggreensfunction}
\ee
where the $\delta$-distribution on curved space-time is defined by
\be
	\int d^{1+d} x' \, \delta^{(1+d)}(x,x') f(x') = f(x)
\ee
for every scalar test function $f(x)$.\\
In the integral equation \eqref{eq:inteq}, we then replace the Minkowski Green's functions $G_k(x_k-x_k')$ by the respective Green's functions $G_k(x_k,x_k')$ on curved spacetime.
\item[(d)] \textit{Interaction kernel.} As mentioned before, a very natural interaction kernel is the symmetric Green's function of the massless KG equation. From the previous points it is clear that this choice possesses an immediate generalization to curved spacetime. Apart from this natural choice, one can also use different interaction kernels, corresponding to different kinds of interaction. It is, however, important that the interaction kernel is formulated in a manifestly covariant way.
\end{enumerate}
With these ingredients, it is clear what \eqref{eq:inteq} means for curved spacetimes. The next step is to find an explicit formulation for certain spacetimes that feature a Big Bang singularity. This mainly requires determining the Green's functions explicitly.

%%%%%
\subsection{FLRW spacetimes} \label{sec:flrw}

Perhaps the best-known spacetimes which feature a Big Bang singularity are the \textit{Friedmann-Lema\^itre-Robertson-Walker (FLRW) spacetimes} (see e.g. \cite[chap. 5.3]{hawking_ellis}). They describe homogeneous and isotropic universes. We shall consider these spacetimes for the rest of the paper. In 1+$d$ spacetime dimensions with $d=2,3$,\footnote{For $d=1$, this form reduces to $ds^2 = a^2(\eta)[d\eta^2 - dz^2]$.} the metric takes the form
\be
	ds^2 = a^2(\eta) \left[ d \eta^2 - dr^2 - s^2_k(r) d \Omega^2_{d-1} \right],
	\label{eq:conformaltimeflrw}
\ee
where $\eta$ is conformal time and $\Omega^2_d$ denotes the canonical metric of the unit sphere $S^d$ in $\R^{d+1}$. There are three cases for the function $s^2_k(r)$:
\be
	s^2_k(r) = \left\{ \begin{array}{cl} \sinh^2 (r),&{\rm open} ~(k=-1)\\  r^2,&{\rm flat}~ (k=0)\\  \sin^2(r),&{\rm closed}~ (k=1).\end{array} \right.
\ee

The function $a(\eta)$ is called the \textit{scale function}. In principle, it could be an arbitrary function of $\eta$. However, for plausible models of matter (such as ``dust'' and ``radiation''), $a(\eta)$ is determined by the Einstein equations. The following table contains some examples (cf.\ Tab.~5.1 in \cite{Wald}).
\begin{table}[h]
	\centering
  	\begin{tabular}{  | c | c | c |}
    		\hline
			&&\\[-3mm]
    		$k$ & dust & radiation \\[2mm] \hline \hline
			&&\\[-3mm]
    		$-1$ &  $\cosh \eta -1$ & $\sinh \eta$ \\[2mm]
    		\hline
			&&\\[-3mm]
    		0 & $\eta^2$ & $|\eta|$ \\[2mm] \hline
    		&&\\[-3mm]	
			1 & $1-\cos \eta$ & $\sin \eta$ \\[2mm] \hline
  	\end{tabular}
 	\caption{Scale factor $a(\eta)$ (up to constant prefactors) in 1+3 spacetime dimensions}
	\label{tab:scalefactor}
\end{table}

\noindent One can see that in the flat and open cases, the scale factor has one exactly one root. This corresponds to a universe which starts with a Big Bang and expands forever thereafter. In the closed case, there is also a Big Crunch in addition to the Big Bang, corresponding to the next root at $\eta = 2\pi$ or $\eta = \pi$, respectively.
The physically relevant spacetime topologies in the three cases are: open $[0,\infty) \times \R^3$, flat $[0,\infty) \times \R^3$, closed $[0,1] \times S^3$. The exact form of $a(\eta)$ will not be relevant for our results, as we can cover whole classes of functions $a(\eta)$, in fact for any dimension $d=1,2,3$.

With this information, we are now ready to turn to calculating the Green's functions.

%%%%%
\subsection{Green's functions of the massless KG equation on FLRW spacetimes} \label{sec:greensfnscurved}

It is, in general, difficult to calculate Green's functions of covariant wave equations explicitly, and only a few examples are known (see, e.g., \cite{haas_poisson}). However, the task is simpler in case the equation under consideration is invariant under conformal transformations, and the spacetime  $(\widetilde{\mathcal{M}}, \widetilde{g}_{ab})$ on which the Green's functions are to be determined is conformally equivalent to another spacetime $(\mathcal{M}, g_{ab})$ on which we already know the Green's functions of the same wave equation. 
Following \cite{john}, we shall explain the method available for this case using the example of the massless KG equation \eqref{eq:curvedkg}, which in fact is conformally invariant (up to appropriate weight factors).\footnote{The same holds true for the massless Dirac equation, see e.g. \cite[chap. 5.7]{penrose_rindler}.}

Let the metrics of the two spacetimes $\mathcal{M}$ and $\widetilde{\mathcal{M}}$ be related by a conformal factor $\Omega(x)$, i.e.,
\be
	\widetilde{g}_{ab} = \Omega^2 \, g_{ab}.
\ee
Then the following relation holds between massless KG operator $(\square_g - \xi R)$ on $(\mathcal{M},g_{ab})$ and the massless the KG operator $(\square_{\widetilde{g}}- \xi\widetilde{R})$ on $(\widetilde{\mathcal{M}},\widetilde{g}_{ab})$ (see \cite{silenko} and \cite[eq. (5.7.19)]{penrose_rindler}):
\be
	\left( \square_g - \xi R \right) ~=~ \Omega^{\frac{d+3}{2}} \left( \square_{\widetilde{g}}- \xi\widetilde{R} \right) \Omega^{-\frac{d-1}{2}}.
	\label{eq:kgconformaltrafo}
\ee
Consequently, if $\phi$ satisfies $( \square_g - \xi R ) \phi = 0$, then $\widetilde{\phi} = \Omega^{\frac{d-1}{2}} \phi$ fulfills $( \square_{\widetilde{g}}- \xi \widetilde{R} ) \widetilde{\phi} = 0$.

Eq.~\eqref{eq:kgconformaltrafo} allows us to deduce how the Green's functions transform under conformal transformations (see \cite[eq. (8)]{john}). Let $G(x,x')$ be a Green's function of $(\square_g- \xi R )$ in the sense of \eqref{eq:kggreensfunction}. Then
\be
	\widetilde{G}(x,x') ~=~ \Omega^{-\frac{d-1}{2}}(x) \, \Omega^{-\frac{d-1}{2}}(x')\, G(x,x')
	\label{eq:greensfntrafo}
\ee
is a Green's function of $( \square_{\widetilde{g}}- \xi \widetilde{R} )$, i.e.,
\be
	\left( \square_{\widetilde{g}}- \xi \widetilde{R} \right) \widetilde{G}(x,x') = [-\widetilde{g}(x)]^{-1/2} \,\delta^{(1+d)}(x,x').
\ee
This can easily be verified using \eqref{eq:kgconformaltrafo}. Note that for $d = 1$, we have that $\widetilde{\phi} = \phi$ and $\widetilde{G} = G$.

In the case with mass, the method still works but is not very useful as the two mass terms are related by:
\be
	\widetilde{m} = \Omega^{-1} \, m,
\ee
independently of $d$. That means, not both mass terms $m, \widetilde{m}$ can be constant.

We shall now employ this method to explicitly calculate the Green's functions of \eqref{eq:curvedkg} on certain FLRW spacetimes in the massless case.

%%%
\paragraph{Flat FLRW spacetimes.}

In the flat case, one can see from \eqref{eq:conformaltimeflrw} that the FLRW spacetimes are conformally equivalent to Minkowski spacetime, with $\Omega(x) = a(\eta)$. Thus, Eq. \eqref{eq:greensfntrafo} allows us to deduce the massless Green's functions of \eqref{eq:curvedkg} from Tab. \ref{tab:greensfnsmink}. The results are summarized in Tab. \ref{tab:greensfnsflatflrw}. We specify the symmetric Green's function $G^\sym(x,x')$. The retarded Green's function $G^\ret(x,x')$ is then given by $G^\ret(\eta,\vx, \eta',\vx') = H(\eta-\eta') \, G^\sym(\eta,\vx, \eta',\vx')$. Moreover, we use the abbreviation $x^2 = \eta^2 - |\vx|^2$.
\begin{table}[h]
\centering
\begin{tabular}{|l|l|}
\hline
	&\\[-3mm]
	$d$ & $G^\sym(x,x')$\\[2mm]
	\hline \hline
	&\\[-3mm]
	1 & $\frac{1}{2} H((x-x')^2)$\\[2mm]
	\hline
	&\\[-3mm]
	2 & $\frac{1}{2\pi} \frac{1}{[a(\eta) a(\eta')]^{1/2}} \frac{H((x-x')^2)}{\sqrt{(x-x')^2}}$\\[2mm]
	\hline
	&\\[-3mm]
	3 & $\frac{1}{2\pi} \frac{1}{a(\eta) a(\eta')} \delta((x-x')^2)$\\[2mm]
	\hline
\end{tabular}
\caption{Symmetric Green's functions of the massless KG equation on a flat FLRW universe}
\label{tab:greensfnsflatflrw}
\end{table}

%%%
\paragraph{Open FLRW spacetime for $d=3$.}
The open FLRW spacetime is globally conformally equivalent to a static spacetime with metric
\be
	d s^2 = d\eta^2 - dr^2 - \sinh^2 r \, d \Omega^2_2.
\ee
This makes it possible to calculate the Green's functions in a similar way as above, with the result \cite[eq. (9.1)]{alertz}
\be
	G^{\ret/\adv}(\eta,\vx,\eta',\vx') = \frac{1}{a(\eta) a(\eta')} \frac{\delta(\eta-\eta'\mp s(\vx,\vx'))}{4\pi \sinh(s(\vx,\vx'))},
\ee
where $s(\vx,\vx')$ is the geodesic distance of the points $\vx, \vx'$ in the hyperbolic 3-space $\mathbb{H}^3$ with metric $dr^2 + \sinh^2 r \, d \Omega^2_2$.

%%%
\paragraph{Closed FLRW spacetime for $d=3$.}

In the closed case, the easiest way to calculate the Green's functions is to exploit the conformal equivalence of the closed FLRW universe with (a part of) the Einstein static universe (ESU). The latter has the topology $\R \times S^3$ and the metric
\be
	ds^2 = d\eta^2 - r_0^2(dr^2 + \sin^2 r \, d \Omega^2_2).
\ee
One can see from \eqref{eq:conformaltimeflrw} that the closed FLRW universe is conformally equivalent to the part of the ESU with $r_0=1$ and $\eta$ taking values between the first and the second root of $a(\eta)$ (such as $\eta \in [0,2\pi]$ for dust and $\eta \in [0,\pi]$ for radiation). The conformal factor is again simply given by $\Omega(x) = a(\eta)$.

The explicit form of the Green's functions of the massless KG equation \eqref{eq:curvedkg} on the ESU has been calculated in the literature. For example, in \cite[eq. (7)]{dowker_76} (see also \cite{dowker_71}) the massless Feynman propagator is stated as:
\be
	G_F(x,x') = \frac{i}{4 \pi^2 r_0} \sum_{n=-\infty}^\infty \frac{s(q,q')+2\pi n r_0}{\sin(s(q,q')/r_0)} \frac{1}{(\eta-\eta')^2 - (s(q,q')+2\pi n r_0)^2 - i\varepsilon}
	\label{eq:dowkergreensfn}
\ee
Here, we denote points $x \in \R \times r_0 S^3$ as $x = (\eta,q)$ with $q \in r_0 S^3$; $s(q,q')$ stands for the geodesic distance of $q,q'$ on $r_0 S^3$. Furthermore,
\be
	\frac{1}{x - i \varepsilon} = {\rm P} \frac{1}{x} + i \pi \, \delta(x),
\ee
where ${\rm P}$ denotes the principal value. Note also that the infinite sum in \eqref{eq:dowkergreensfn} results from the fact that in the ESU images can travel around the spatial dimensions of the universe and arrive back at the same spatial location.

We are interested in Green's functions different from the Feynman propagator, as the latter does not vanish outside of the light cone.
According to \cite[p. 82]{fulling}, the Feynman propagator is related to the symmetric Green's function $G^\sym(x,x')$ as follows:
\be
	G^\sym(x,x') = \mathrm{Re} \, G_F(x,x').
\ee
For the case of \eqref{eq:dowkergreensfn}, this leads to:
\be
	G^\sym(x,x') = -\frac{1}{4 \pi r_0} \sum_{n=-\infty}^\infty \frac{s+2\pi n r_0}{\sin(s/r_0)}   \, \delta \left( (\eta-\eta')^2 - (s+2\pi n r_0)^2\right).
	\label{eq:gsymesu}
\ee
Thus, from \eqref{eq:greensfntrafo} we conclude that the symmetric Green's function of the closed FLRW universe is given by (denoting spacetime points by $x = (\eta,q)$ with $q \in S^3$ and now setting $r_0=1$):
\be
	G^\sym_{{\rm FLRW},k=1}(x,x') = -\frac{1}{4\pi}  \frac{1}{a(\eta) a(\eta')}\sum_{n=-\infty}^\infty \frac{s(q,q')+2\pi n}{\sin (s(q,q'))}   \, \delta \left( (\eta-\eta')^2 - (s(q,q')+2\pi n)^2\right).
	\label{eq:closedgreensfn}
\ee
Note that we have $\eta, \eta' \in [0,T]$ for some $T>0$. Thus, only finitely many terms in the sum are non-zero. For example, for a universe filled with dust, we have that $T=2\pi$ (see Tab. \ref{tab:scalefactor}). Then, noting that $0\leq s(q,q') \leq \pi$, one can see that the only terms which contribute to $G^\sym_{{\rm FLRW},k=1}$ are the ones with $n=0$ and $n=-1$. In general, the sum contains at most $\left\lfloor \frac{T}{\pi}\right\rfloor +1$ terms.

With these preparations, we are now ready to give a number of explicit examples for the integral equation \eqref{eq:inteq} on FLRW spacetimes.

%%%%%
\subsection{Explicit examples for the integral equation on FLRW spacetimes} \label{sec:explicitexamples}

In the following, we formulate the integral equation for two different cases. (a) Flatand (b) open FLRW spacetimes in the case of retarded Green's functions, as well as (c) closed FLRW spacetimes in the case of symmetric Green's functions. In both cases, we consider massless scalar particles. The reasons for focusing on theses cases are the following. First and foremost, these are cases in which the time interval of integration becomes finite (see \cite[Sec. 2.2]{mtve} for a detailed illustration of the problems that can arise with infinite intervals). Case (a) is similar to the one studied previously in \cite{mtve}, see Sec. \ref{sec:previousresults}, and ensures that the integral equation \eqref{eq:inteq} naturally attains a Volterra structure \eqref{Volterra} in the time variables. Compared to Minkowski spacetime, the retarded case seems more natural for the flat and open FLRW universes which feature a Big Bang but not a Big Crunch, i.e., an intrinsic asymmetry between past and future. For closed FLRW spacetimes on the other hand (case (c)), there exist both a Big Bang and a Big Crunch. The universe then does not have any preferred time direction, and the case of symmetric Green's functions appears most natural.

%%%
\subsubsection{Massless scalar particles on flat FLRW spacetimes in the retarded case}

For the spacetime volume elements of flat FLRW spacetime, we obtain, using coordinates $x = (\eta,\vx)$ with $\eta \in [0,\infty)$ and $\vx \in \R^d$:
\be
	dV(x) = a^{1+d}(\eta) \, d \eta \, d^d \vx.
\ee
Thus, using the Green's functions in Tab. \ref{tab:greensfnsflatflrw} we can write down the integral equation for $d=1,2,3$ as follows.

%%%
\paragraph{d=1:}
\begin{align}
	&\psi(\eta_1,z_1,\eta_2,z_2) = \psi^\free(\eta_1,z_1,\eta_2,z_2) + \frac{\lambda}{8} \int_0^{\eta_1} d\eta_1' \int_0^{\eta_2} d\eta_2' \int dz_1' \, dz_2'~a^2(\eta_1') \, a^2(\eta_2')\nonumber\\
&\times\,H(\eta_1-\eta_1'-|z_1-z_1'|)\,  H(\eta_2-\eta_2'-|z_2-z_2'|) \, H((\eta_1'-\eta_2')^2-|z_1'-z_2'|^2) \, \psi(\eta_1',z_1',\eta_2',z_2').
\label{eq:inteq1dflatflrw}
\end{align}
Here, the physically natural interaction kernel, given by the symmetric Green's function, $K(\eta_1,z_1,\eta_2,z_2) = \frac{1}{2} H((\eta_1-\eta_2)^2-|z_1-z_2|^2)$, is bounded. Moreover, the Big Bang singularity does not lead to any singularities in the integral equation.

We shall also study the integral equation with a general bounded interaction kernel $\widetilde{K}(\eta_1,z_1,\eta_2,z_2)$ instead of $K$:
\begin{align}
	&\psi(\eta_1,z_1,\eta_2,z_2) = \psi^\free(\eta_1,z_1,\eta_2,z_2) + \frac{\lambda}{4} \int_0^{\eta_1} d\eta_1' \int_0^{\eta_2} d\eta_2' \int dz_1' \, dz_2'~a^2(\eta_1') \, a^2(\eta_2')\nonumber\\
&\times\,H(\eta_1-\eta_1'-|z_1-z_1'|)\,  H(\eta_2-\eta_2'-|z_2-z_2'|) \, \widetilde{K}(\eta_1',z_1',\eta_2',z_2') \, \psi(\eta_1',z_1',\eta_2',z_2').
\label{eq:inteq1dflatflrwgeneralk}
\end{align}

%%%
\paragraph{d=2:} We obtain
\begin{align}
	\psi(\eta_1,\vx_1,\eta_2,\vx_2) = ~&\psi^\free(\eta_1,\vx_1,\eta_2,\vx_2) + \frac{\lambda}{(2\pi)^3} \frac{1}{[a(\eta_1) a(\eta_2)]^{1/2}} \int_0^{\eta_1} d\eta_1' \int_0^{\eta_2} d\eta_2' \int d^2 \vx_1'\, d^2 \vx_2'\nonumber\\
& ~\times~ a^2(\eta_1') \, a^2(\eta_2') ~\frac{H(\eta_1-\eta_1'-|\vx_1-\vx_1'|)}{\sqrt{(\eta_1-\eta_1')^2-|\vx_1-\vx_1'|^2}} \, \frac{H(\eta_2-\eta_2'-|\vx_2-\vx_2'|)}{\sqrt{(\eta_2-\eta_2')^2-|\vx_2-\vx_2'|^2}}\nonumber\\
& ~\times~ \frac{H((\eta_1'-\eta_2')^2-|\vx_1'-\vx_2'|^2)}{\sqrt{(\eta_1'-\eta_2')^2-|\vx_1'-\vx_2'|^2}} \,  \psi(\eta_1',\vx_1',\eta_2',\vx_2').
\label{eq:inteq2dflatflrw}
\end{align}
Here, the natural interaction kernel
\be
	K(\eta_1,\vx_1,\eta_2,\vx_2) = \frac{1}{2\pi} \frac{1}{[a(\eta_1) a(\eta_2)]^{1/2}} \frac{H((\eta_1-\eta_2)^2-|\vx_1-\vx_2|^2)}{\sqrt{(\eta_1-\eta_2)^2-|\vx_1-\vx_2|^2}}
	\label{eq:naturalkernel2dflatflrw}
\ee
is singular in two ways: it diverges as $\eta_1\to 0$ or $\eta_2\to 0$ because then the middle denominator tends to 0, and it diverges as $x_2$ approaches the light cone of $x_1$ because then the last denominator tends to 0. In addition, the singular factor $a^{-1/2}(\eta_1) a^{-1/2} (\eta_2)$ appears in front of the integral. We shall see in Sec. \ref{sec:existence} that this leads to a singular behavior of the wave function proportional to that factor for $\eta_1,\eta_2\rightarrow 0$.

We shall consider a simplified version of \eqref{eq:inteq2dflatflrw} where the natural interaction kernel \eqref{eq:naturalkernel2dflatflrw} is replaced by
 $a^{-1/2}(\eta_1) a^{-1/2}(\eta_2) \widetilde{K}(\eta_1,\vx_1,\eta_2,\vx_2)$, where $\widetilde{K}$ is a bounded function.
(We keep the singular prefactors resulting from the conformal transformation, as they cancel with the some of the conformal factors in the spacetime volume elements.)
This leads us to the simplified equation
\begin{align}
	\psi(\eta_1,\vx_1,\eta_2,\vx_2) =\ &\psi^\free(\eta_1,\vx_1,\eta_2,\vx_2) + \frac{\lambda}{(2\pi)^2} \frac{1}{[a(\eta_1) a(\eta_2)]^{1/2}} \int_0^{\eta_1} d\eta_1' \int_0^{\eta_2} d\eta_2' \int d^2 \vx_1'\, d^2 \vx_2'\nonumber\\
& ~\times~ a^2(\eta_1') \, a^2(\eta_2') ~\frac{H(\eta_1-\eta_1'-|\vx_1-\vx_1'|)}{\sqrt{(\eta_1-\eta_1')^2-|\vx_1-\vx_1'|^2}} \, \frac{H(\eta_2-\eta_2'-|\vx_2-\vx_2'|)}{\sqrt{(\eta_2-\eta_2')^2-|\vx_2-\vx_2'|^2}}\nonumber\\
& ~\times~ \widetilde{K}(\eta_1',\vx_1',\eta_2',\vx_2') \,  \psi(\eta_1',\vx_1',\eta_2',\vx_2').
\label{eq:inteq2dflatflrwgeneralk}
\end{align}

%%%
\paragraph{d=3:}
\begin{align}
	\psi(\eta_1,\vx_1,\eta_2,\vx_2) =\ & \psi^\free(\eta_1,\vx_1,\eta_2,\vx_2) + \frac{2\lambda}{(4\pi)^3} \frac{1}{a(\eta_1) a(\eta_2)} \int_0^{\eta_1} d\eta_1' \int d^3 \vx_1' \int_0^{\eta_2} d\eta_2' \int d^3 \vx_2'~  \nonumber\\
& ~\times~ a^2(\eta_1') \, a^2(\eta_2') \,\frac{\delta(\eta_1-\eta_1'-|\vx_1-\vx_1'|)}{|\vx_1-\vx_1'|}\, \frac{\delta(\eta_2-\eta_2'-|\vx_2-\vx_2'|)}{|\vx_2-\vx_2'|} \nonumber\\
& ~\times~\delta((\eta_1'-\eta_2')^2-|\vx_1'-\vx_2'|^2)\,  \psi(\eta_1',\vx_1',\eta_2',\vx_2').
\label{eq:inteq3dflatflrw}
\end{align}
Again, there is a singular prefactor in front of the integral. Furthermore, the interaction kernel 
\be
K(\eta_1,\vx_1,\eta_2,\vx_2) = \frac{1}{2\pi} a^{-1}(\eta_1) a^{-1}(\eta_2) \delta((\eta_1-\eta_2)^2-|\vx_1-\vx_2|^2)
\ee
is singular because of the $a^{-1}$ factors and because of the $\delta$-function. We shall therefore consider a simplified version of \eqref{eq:inteq3dflatflrw} where $K$ is replaced by $a^{-1}(\eta_1) a^{-1}(\eta_2) \widetilde{K}(\eta_1,\vx_1,\eta_2,\vx_2)$ where $\widetilde{K}$ is a bounded function. Formally integrating out the remaining $\delta$-functions then leads to:
\begin{align}
	&\psi(\eta_1,\vx_1,\eta_2,\vx_2) = \psi^\free(\eta_1,\vx_1,\eta_2,\vx_2) + 
\frac{\lambda}{(4\pi)^2}
	\frac{1}{a(\eta_1) a(\eta_2)} \int d^3 \vx_1' \, d^3 \vx_2'~  \nonumber\\
& ~\times~a^2(\eta_1-|\vx_1-\vx_1'|) \, a^2(\eta_2-|\vx_2-\vx_2'|)  \,\frac{H(\eta_1-|\vx_1-\vx_1'|)}{|\vx_1-\vx_1'|}\, \frac{H(\eta_2-|\vx_2-\vx_2'|)}{|\vx_2-\vx_2'|}\nonumber\\
& ~\times~ \widetilde{K}(\eta_1-|\vx_1-\vx_1'|,\vx_1',\eta_2-|\vx_2-\vx_2'|,\vx_2') \,\psi(\eta_1-|\vx_1-\vx_1'|,\vx_1',\eta_2-|\vx_2-\vx_2'|,\vx_2').
\label{eq:inteq3dflatflrwgeneralk}
\end{align}
Here, the Heaviside functions result from the lower limits $0 \leq \eta_i'$ in \eqref{eq:inteq3dflatflrw}.

%%%
\paragraph{Remark.} Note that in the integral equations \eqref{eq:inteq1dflatflrw}, \eqref{eq:inteq2dflatflrw}, and \eqref{eq:inteq3dflatflrw}, the $\eta_i'$-integrals run only from 0 to $\eta_i$. That is, we have obtained a Volterra structure as in \eqref{Volterra} in the time variables, which makes it possible to use the same methods for the existence and uniqueness proofs as in \cite{mtve} (see Sec. \ref{sec:existence}).

%%%
\subsubsection{Massless scalar particles on open FLRW spacetimes in the retarded case}
Here we consider only the case $d=3$ as it is the physically relevant case.
As in the flat case, we assume that $a(\eta)$ is a continuous function  with $a(0)=0$ and $a(\eta)>0$ for $\eta>0$. We also replace the physically natural interaction kernel
\be
K(\eta_1,\vx_1,\eta_2,\vx_2) = \frac{a^{-1}(\eta_1) a^{-1}(\eta_2) }{4\pi \sinh(s(\vx_1,\vx_2))} \left[\delta(\eta_1-\eta_2-s(\vx_1,\vx_2)) +  \delta(\eta_1-\eta_2+s(\vx_1,\vx_2)) \right]
\ee
by
\be
a^{-1}(\eta_1) a^{-1}(\eta_2) \widetilde{K}(\eta_1,\vx_1,\eta_2,\vx_2)\,,
\ee
where $\widetilde{K}$ is either a bounded function or only mildly singular. After using the remaining $\delta$-functions in the Green's functions $G_1^\ret, G_2^\ret$ to remove the time integrals, our integral equation reads:
\begin{align}
	&\psi(\eta_1,\vx_1,\eta_2,\vx_2) =  \psi^\free(\eta_1,\vx_1,\eta_2,\vx_2) + \frac{\lambda}{(4\pi)^2} \frac{1}{a(\eta_1) a(\eta_2)} \int d^3 \vx_1' \int d^3 \vx_2'~  \nonumber\\
& ~\times~ a^2(\eta_1') \, a^2(\eta_2') \,\frac{H(\eta_1-\eta_1'-s(\vx_1,\vx_1'))}{\sinh ( s(\vx_1,\vx_1'))}\, \frac{H(\eta_2-\eta_2'-s(\vx_2,\vx_2'))}{\sinh (s(\vx_2,\vx_2'))} \nonumber\\
& ~\times~\widetilde{K}(\eta_1-s(\vx_1,\vx_1'),\vx_1',\eta_2-s(\vx_2,\vx_2'),\vx_2')\,  \psi(\eta_1-s(\vx_1,\vx_1'),\vx_1',\eta_2-s(\vx_2,\vx_2'),\vx_2').
\label{eq:inteq3dopenflrw}
\end{align}
Here, $d^3\vx$ denotes the infinitesimal surface element of $\mathbb{H}^3$ and $s(\vx,\vx')$ stands for the geodesic distance of $\vx,\vx' \in \mathbb{H}^3$.

%%%
\paragraph{Remark.} As in the flat case, the domains of integration in \eqref{eq:inteq3dopenflrw} are limited by $\eta_1,\eta_2$. The equation thus has a Volterra feature. Apart from the replacement of $\R^3$ with $\mathbb{H}^3$, Eqs. \eqref{eq:inteq3dflatflrwgeneralk} and \eqref{eq:inteq3dopenflrw} are qualitatively very similar.

%%%%%
\subsubsection{Massless scalar particles on closed FLRW spacetimes in the time-symmetric case}
Here we consider only $d=3$ as well.
We assume that the scale function $a(\eta)$ is a continuous function with $a(0) = 0$ (corresponding to the Big Bang) and $a(T) = 0$ for some $T>0$ (corresponding to the Big Crunch) while $a(\eta) > 0$ for $\eta \in (0,T)$. Furthermore, we denote points $x$ in the closed FLRW universe by coordinates $(\eta,q)$ with $\eta \in [0,T]$ and $q \in S^3$.

Then our integral equation reads:
\begin{align}
	&\psi(\eta_1,q_1,\eta_2,q_2) = \psi^\free(\eta_1,q_1,\eta_2,q_2) - \frac{\lambda}{(4\pi)^3} \frac{1}{a(\eta_1) a(\eta_2)} \int_0^T d\eta_1' \int_0^T d\eta_2' \int d\Omega_3(q_1')\, d\Omega_3(q_2') \nonumber\\
& ~\times~ a^2(\eta_1') \, a^2(\eta_2') \sum_{l,m,n =-\infty}^\infty \frac{s(q_1,q_1')+2\pi l}{\sin (s(q_1,q_1'))}   \, \delta \left( (\eta_1-\eta_1')^2 - (s(q_1,q_1')+2\pi l)^2\right)\nonumber\\
&~\times~ \frac{s(q_2,q_2')+2\pi m}{\sin (s(q_2,q_2'))} \, \delta \left( (\eta_2-\eta_2')^2 - (s(q_2,q_2')+2\pi m)^2\right)\nonumber\\
&~\times~ \frac{s(q_1',q_2')+2\pi n}{\sin (s(q_1',q_2'))} \, \delta \left( (\eta_1'-\eta_2')^2 - (s(q_1',q_2')+2\pi n)^2\right) \, \psi(\eta_1',q_1',\eta_2',q_2').
\label{eq:inteq3dclosedflrw}
\end{align}
Here, $d\Omega_3$ denotes the infinitesimal surface element of $S^3$ and $s(q,q')$ stands for the geodesic distance of $q,q' \in S^3$.

The interaction kernel of \eqref{eq:inteq3dclosedflrw} is highly singular. As before, we shall therefore consider a simplified problem where the interaction kernel
\be
K(\eta_1,q_1,\eta_2,q_2) = \frac{-(4\pi)^{-1}}{a(\eta_1) a(\eta_2)} \sum_n \frac{s(q_1,q_2)+2\pi n}{\sin(s(q_1,q_2))} \delta((\eta_1-\eta_2)^2- (s(q_1,q_2)+2\pi n)^2)
\ee
is replaced by 
\be
a^{-1}(\eta_1) a^{-1}(\eta_2) \widetilde{K}(\eta_1,q_1,\eta_2,q_2)\,,
\ee
where $\widetilde{K}$ is either a bounded function or a function with only a specific type of singularity (see Thm. \ref{thm:existenceclosedsingular} for details).
We decompose the remaining $\delta$-functions according to
\begin{align}
	& \delta((\eta_i-\eta_i')^2- (s(q_i,q_i')+2\pi n)^2)\nonumber\\
&= \frac{1}{2 |s(q_i,q_i')+2\pi n|} \big[ \delta(\eta_i-\eta_i'-|s(q_i,q_i')+2\pi n|) + \delta(\eta_i-\eta_i'+|s(q_i,q_i')+2\pi n|)  \big].
\end{align}
Integrating out the $\delta$-functions then leads to
\begin{align}
	&\psi(\eta_1,q_1,\eta_2,q_2) = \psi^\free(\eta_1,q_1,\eta_2,q_2) + \frac{\lambda}{4(4\pi)^2} \frac{1}{a(\eta_1) a(\eta_2)}  \sum_{l,m =-\infty}^\infty  \int d\Omega_3(q_1')\, d\Omega_3(q_2') \nonumber\\
& ~\times~\frac{\sgn(s(q_1,q_1')+2\pi l)}{\sin (s(q_1,q_1'))}  \, \frac{\sgn(s(q_2,q_2')+2\pi m)}{\sin (s(q_2,q_2'))}\nonumber\\
& ~\times~ \sum_{\sigma_1, \sigma_2 = \pm 1} \Big[ \id_{[0,T]}(\eta_1+\sigma_1|s(q_1,q_1')+2\pi l|) \, \id_{[0,T]}(\eta_2+\sigma_2 |s(q_2,q_2')+2\pi m|) \nonumber\\
& ~\times~ a^2(\eta_1+\sigma_1|s(q_1,q_1')+2\pi l|)  \, a^2(\eta_2+\sigma_2|s(q_2,q_2')+2\pi m|)\nonumber\\
& ~\times~(\widetilde{K}\times \psi)(\eta_1+\sigma_1|s(q_1,q_1')+2\pi l|,q_1',\eta_2+\sigma_2|s(q_2,q_2')+2\pi m|,q_2') \Big],
\label{eq:inteq3dclosedflrwgeneralk}
\end{align}
where $\id_{[0,T]}$ denotes the indicator function of the interval $[0,T]$.
Note that \eqref{eq:inteq3dclosedflrwgeneralk} has a compact domain of integration. Moreover, as a consequence of $T$ being finite, the sums over $l,m$ contain only finitely many terms (for universes filled with dust these are only the terms with $l,m\in\{-1,0\}$, see Tab. \ref{tab:scalefactor}). This corresponds to the fact that in the closed FLRW spacetimes, light can travel around the spatial dimensions of the universe at most a finite number of times (at most once for dust).

However, note also that contrary to the retarded case for flat and open FLRW spacetimes, we do not obtain a Volterra structure of the equation. This makes it more difficult to prove the existence and uniqueness of solutions.

%%%%%%%%%%%%
%%%%%%%%%%%%
\section{Existence and uniqueness of solutions} \label{sec:existence}

We now prove the existence and uniqueness of solutions for the simplified models in the flat and closed cases. In the case of flat FLRW spacetimes, it is possible to reduce the problem to the one on the Minkowski half-space $\frac{1}{2} \R^{1,d}$ (Sec. \ref{sec:previousresults}). In case of closed FLRW spacetimes, we show that the integral operator in \eqref{eq:inteq3dclosedflrwgeneralk} is bounded. We then obtain the existence and uniqueness of solutions via Banach's fixed point theorem for small coupling constants $\lambda$.
The reason for not considering the open case here is that it is qualitatively similar to the flat case so that no great new insights are expected. At the same time the existence and uniqueness of solutions cannot be directly reduced to the results for a Minkowski half-space. One would have to redo large parts of the proofs in \cite{mtve} in a similar but slightly different way, replacing $\R^3$ with $\mathbb{H}^3$. We try to avoid a duplication here.

%%%%%
\subsection{Flat FLRW spacetimes in the retarded case} \label{sec:existenceflat}

We now show how to reduce the problem on flat FLRW spacetimes to the one described in Sec. \ref{sec:previousresults}.
First, we note that the transformation behavior \eqref{eq:kgconformaltrafo} of the KG operator under conformal transformations implies that we can obtain a solution $\widetilde{\phi}$ of the free KG equation on flat FLRW spacetime from a solution $\phi$ of that equation on $\frac{1}{2}\R^{1,d}$ (and vice versa) by
\be
	\widetilde{\phi} = a^{-\frac{d-1}{2}}(\eta) \, \phi.
\ee
This means that $\widetilde{\phi}$ diverges at the Big Bang (except for $d=1$). We therefore expect solutions of the integral equations \eqref{eq:inteq1dflatflrwgeneralk}, \eqref{eq:inteq2dflatflrwgeneralk}, and \eqref{eq:inteq3dflatflrwgeneralk} to be proportional to $a^{-\frac{d-1}{2}}(\eta_1) \,a^{-\frac{d-1}{2}}(\eta_2)$ for $\eta_1, \eta_2\rightarrow 0$. The precise statement reads as follows.

%%%
\begin{theorem} \label{thm:flatflrw}
	Let $T>0$, $d \in \{ 1,2,3\}$. Furthermore, let $a: [0,\infty) \rightarrow \infty$ be a continuous function with $a(0) = 0$ and $a(\eta)>0$ for $\eta > 0$, and $\widetilde{K} : ([0,\infty) \times \R^d)^2 \rightarrow \C$ be bounded.\\
 Then for every $\psi^\free$ with $a^{\frac{d-1}{2}}(\eta_1)\, a^{\frac{d-1}{2}}(\eta_2) \, \psi^\free \in \Banach_d$, the respective integral equation on the (1+$d$)--dimensional flat FLRW universe with scale function $a(\eta)$ (one of the integral equations \eqref{eq:inteq1dflatflrwgeneralk}, \eqref{eq:inteq2dflatflrwgeneralk}, or \eqref{eq:inteq3dflatflrwgeneralk}) has a unique solution $\psi$ with $a^{\frac{d-1}{2}}(\eta_1) \,a^{\frac{d-1}{2}}(\eta_2) \, \psi \in \Banach_d$ for $0 \leq \eta_1, \eta_2 \leq T$. ($\Banach_d$ is defined in \eqref{eq:banach}.)
\end{theorem}

\begin{proof}
	The idea is to reduce the statement to Thm.~\ref{thm:boundedkernels}. We consider each $d$ separately.

For $d=1$, we have $\frac{d-1}{2}=0$ and the integral equation \eqref{eq:inteq1dflatflrwgeneralk} is already of the form \eqref{eq:inteq1dsimplified} with $t_i \leftrightarrow \eta_i$ and $K(t_1,z_1,t_2,z_2) = a^2(t_1) a^2(t_2) \widetilde{K}(t_1,z_1,t_2,z_2)$. This kernels is bounded for $0 \leq t_1,t_2 \leq T$. Hence, Thm. \ref{thm:boundedkernels} yields the claim.

For $d=2$, we have $\frac{d-1}{2} = \frac{1}{2}$. Multiplying \eqref{eq:inteq2dflatflrwgeneralk} with  $a^{1/2}(\eta_1) a^{1/2}(\eta_2)$ and introducing $\chi = a^{1/2}(\eta_1) a^{1/2}(\eta_2) \psi$ as well as $\chi^\free = a^{1/2}(\eta_1) a^{1/2}(\eta_2) \psi^\free$ yields
\begin{align}
	\chi(\eta_1,\vx_1,\eta_2,\vx_2) = \ & \chi^\free(\eta_1,\vx_1,\eta_2,\vx_2) + \frac{\lambda}{(2\pi)^2} \int_0^{\eta_1} d\eta_1' \int_0^{\eta_2} d\eta_2' \int d^2 \vx_1'\, d^2 \vx_2'\nonumber\\
& \times~ a^{3/2}(\eta_1') \, a^{3/2}(\eta_2') ~\frac{H(\eta_1-\eta_1'-|\vx_1-\vx_1'|)}{\sqrt{(\eta_1-\eta_1')^2-|\vx_1-\vx_1'|^2}} \, \frac{H(\eta_2-\eta_2'-|\vx_2-\vx_2'|)}{\sqrt{(\eta_2-\eta_2')^2-|\vx_2-\vx_2'|^2}}\nonumber\\
& \times~ \widetilde{K}(\eta_1',\vx_1',\eta_2',\vx_2') \,  \chi(\eta_1',\vx_1',\eta_2',\vx_2').
\label{eq:inteq2dflatflrwgeneralk2}
\end{align}
So if $\psi$ solves \eqref{eq:inteq2dflatflrwgeneralk} for some $\psi^\free$, then $\chi$ solves \eqref{eq:inteq2dflatflrwgeneralk2} for $\chi^\free$. The converse also holds.\\
Now, \eqref{eq:inteq2dflatflrwgeneralk2} has the same form as \eqref{eq:inteq2dsimplified} with $K(t_1,\vx_1,t_2,\vx_2) = a^{3/2}(t_1) a^{3/2}(t_2) \widetilde{K}(t_1,\vx_1,t_2,\vx_2)$ and $m_1=m_2=0$. Thus, Thm. \ref{thm:boundedkernels} yields the claim.

For $d=3$, we have $\frac{d-1}{2} = 1$. Similarly as for $d=2$, multiplying \eqref{eq:inteq3dflatflrwgeneralk} by  $a(\eta_1) a(\eta_2)$ and introducing $\chi = a(\eta_1) a(\eta_2) \psi$ as well as $\chi^\free = a(\eta_1) a(\eta_2) \psi^\free$ shows that \eqref{eq:inteq3dflatflrwgeneralk} is equivalent to
\begin{align}
	&\chi(\eta_1,\vx_1,\eta_2,\vx_2) = \chi^\free(\eta_1,\vx_1,\eta_2,\vx_2) + \frac{\lambda}{(4\pi)^2} \int d^3 \vx_1' \, d^3 \vx_2'~  \nonumber\\
& ~\times~a(\eta_1-|\vx_1-\vx_1'|) \, a(\eta_2-|\vx_2-\vx_2'|)  \,\frac{H(\eta_1-|\vx_1-\vx_1'|)}{|\vx_1-\vx_1'|}\, \frac{H(\eta_2-|\vx_2-\vx_2'|)}{|\vx_2-\vx_2'|}\nonumber\\
& ~\times~ (\widetilde{K} \times \chi)(\eta_1-|\vx_1-\vx_1'|,\vx_1',\eta_2-|\vx_2-\vx_2'|,\vx_2').
\label{eq:inteq3dflatflrwgeneralk2}
\end{align}
Eq.~\eqref{eq:inteq3dflatflrwgeneralk2} is of the same form as \eqref{eq:inteq3dsimplified} with $K(t_1,\vx_1,t_2,\vx_2) = a(t_1) a(t_2) \widetilde{K}(t_1,\vx_1,t_2,\vx_2)$. Therefore, the claim is again reduced to Thm. \ref{thm:boundedkernels}. \qed
\end{proof}

\paragraph{Remarks.}
\begin{enumerate}
	\item The result covers the physically most natural interaction kernel for $d=1$. For $d=2$ and $d=3$, the physically most natural interaction kernels are singular (in addition to the singularities coming from the scale function) and therefore not covered by the theorem. Remarkably, however, Thm. \ref{thm:flatflrw} covers certain manifestly covariant interaction kernels. A class of examples is given by
\be
	\widetilde{K}(x_1,x_2) = 
	\begin{cases} f(d(x_1,x_2)) & \text{if $x_1,x_2$ are time-like related}\\
	0 & \text{else,}
	\end{cases}
\ee
where $d(x_1,x_2) = (|\eta_1-\eta_2|- |\vx_1-\vx_2|) \int_0^1 a(\tau\eta_1+(1-\tau)\eta_2) d\tau$ is the time-like distance of the spacetime points $x_1= (\eta_1,\vx_1)$ and $x_2 = (\eta_2,\vx_2)$, and $f$ is an arbitrary bounded function.
 The corresponding integral equations \eqref{eq:inteq1dflatflrwgeneralk}, \eqref{eq:inteq2dflatflrwgeneralk}, and \eqref{eq:inteq3dflatflrwgeneralk} define rigorous,  relativistic and interacting quantum dynamics in 1+1, 1+2, and 1+3 spacetime dimensions, respectively.
\item As a consequence of the Volterra structure of Eqs. \eqref{eq:inteq1dflatflrwgeneralk}, \eqref{eq:inteq2dflatflrwgeneralk}, and \eqref{eq:inteq3dflatflrwgeneralk}, we have that as $\eta_1, \eta_2 \rightarrow 0$, $\psi$ is asymptotically equal to $\psi^\free$, i.e.,
\be\label{etato0}
\lim_{\eta_1,\eta_2\to 0}a(\eta_1)a(\eta_2)\psi = 
\lim_{\eta_1,\eta_2\to 0}a(\eta_1)a(\eta_2)\psi^\free\,.
\ee
Moreover, if $\psi^\free$ is a solution of the free multi-time KG equations\eqref{freeKGcurved}, then it is itself determined by Cauchy data, for example by the right-hand side of \eqref{etato0} (along with data for $\partial \psi^\free/\partial \eta_i$), which thus plays the role of \textit{initial data for $\psi$ at the Big Bang.}
\end{enumerate}

%%%
\paragraph{Special singular interaction kernels for $d=3$.}
Besides bounded interaction kernels, we can also treat singular interaction kernels. In a similar way as Thm. \ref{thm:flatflrw} follows from Thm. \ref{thm:boundedkernels}, Thm. \ref{thm:singularkernel3d} implies the following result for $d=3$ and interaction kernels with a $1/|\vx_1-\vx_2|$--singularity:

\begin{theorem} \label{thm:singkernelflatflrw}
	Let $f:([0,\infty) \times \R^3)^2 \rightarrow \C$ be a bounded function. Then, under the same assumptions as in Thm. \ref{thm:flatflrw} for $d=3$ but with
	\be
		\widetilde{K}(\eta_1,\vx_1,\eta_2,\vx_2) = \frac{f(\eta_1,\vx_1,\eta_2,\vx_2)}{|\vx_1-\vx_2|},
	\ee
	the integral equation \eqref{eq:inteq3dflatflrwgeneralk} has a unique solution $\psi$ with $a(\eta_1) a(\eta_2) \psi \in \Banach_3$ for every $\psi^\free$ with $a(\eta_1) a(\eta_2) \psi^\free \in \Banach_3$.
\end{theorem}

%%%%%%
\subsection{Closed FLRW spacetimes in the time-symmetric case} \label{sec:existenceclosed}

For closed FLRW universes, the integral equation \eqref{eq:inteq3dclosedflrwgeneralk} does not have a Volterra structure. Therefore, the techniques we have used for flat FLRW spacetimes are not available. Instead, we show that the integral operator in \eqref{eq:inteq3dclosedflrwgeneralk} is a bounded operator on a suitable Banach space and therefore defines a contraction, provided the coupling constant $\lambda$ is small enough.\footnote{Similar smallness conditions are often needed to obtain the existence of solutions for general Fredholm integral equations (i.e., without a Volterra structure).} Then the existence and uniqueness of solutions follows from Banach's fixed point theorem.

First, note that the transformation behavior \eqref{eq:kgconformaltrafo} implies that $\psi$ now has a singularity proportional to $a^{-1}(\eta_1)a^{-1}(\eta_2)$ for $\eta_1, \eta_2 \rightarrow T$ as well as for $\eta_1, \eta_2 \rightarrow 0$.

We shall prove the following theorem.

%%%%%%
\begin{theorem} \label{thm:existenceclosedsingular}
	Let $T>0$ and $a: [0,T] \rightarrow [0,\infty)$ be a continuous function with $a(0) = 0 = a(T)$ and $a(\eta)>0$ for $\eta\in (0,T)$. Moreover, consider the Banach space
	\be
		\Banach = L^\infty \big( [0,T]^2, L^2((S^3)^2) \big)
	\ee
	with norm $\| \psi\|_\Banach = \esssup_{\eta_1,\eta_2\in [0,T]} \| \psi(\eta_1,\cdot,\eta_2,\cdot)\|_{L^2((S^3)^2)}$, and let $f: \big( [0,T] \times S^3 \big)^2 \rightarrow \C$ be a bounded function.

Then for every $\psi^\free$ with $a(\eta_1) a(\eta_2)\psi^\free \in \Banach$, the integral equation \eqref{eq:inteq3dclosedflrwgeneralk} with
\be
	\widetilde{K}(\eta_1,q_1,\eta_2,q_2) = \frac{f(\eta_1,q_1,\eta_2,q_2)}{\sin(s(q_1,q_2))}
	\label{eq:singularkclosed}
\ee
possesses a unique solution $\psi$ with $a(\eta_1) a(\eta_2)\psi \in \Banach$, provided the coupling constant $\lambda$ satisfies
\be
	|\lambda| < \left( \frac{\pi^2}{\sqrt{2}} \,\left(\left\lfloor \frac{T}{\pi}\right\rfloor +1\right)^2 \, \| a\|^2_\infty \, \| f \|_\infty\right)^{-1}.
	\label{eq:smalllambda2}
\ee
\end{theorem}

\begin{remark}
	The structure of the singularity of $\widetilde{K}$ in \eqref{eq:singularkclosed} imitates the singularity of the Green's function \eqref{eq:closedgreensfn} for the closed FLRW spacetime (with the $\delta$-function replaced by a bounded function $f$). The case of bounded $\widetilde{K}$ is included in Thm. \ref{thm:existenceclosedsingular} for $f(\eta_1,q_1,\eta_2,q_2) \propto \sin(s(q_1,q_2))$. This is different from the case of flat FLRW spacetimes where the $1/|\vx_1-\vx_2|$ singularity cannot be exactly compensated by $|\vx_1-\vx_2|$ as the latter is not bounded on $\R^3\times \R^3$.
\end{remark}

\begin{proof}
	Let $\chi = a(\eta_1) a(\eta_2)\psi$ and  $\chi^\free = a(\eta_1) a(\eta_2)\psi^\free$. Then \eqref{eq:inteq3dclosedflrwgeneralk} is equivalent to
\begin{align}
	&\chi(\eta_1,q_1,\eta_2,q_2) = \chi^\free(\eta_1,q_1,\eta_2,q_2) +\frac{\lambda}{4(4\pi)^2} \sum_{l,m =-\infty}^\infty  \int d\Omega_3(q_1')\, d\Omega_3(q_2') \nonumber\\
& ~\times~\frac{\sgn(s(q_1,q_1')+2\pi l)}{\sin (s(q_1,q_1'))}  \, \frac{\sgn(s(q_2,q_2')+2\pi m)}{\sin (s(q_2,q_2'))}\nonumber\\
& ~\times~ \sum_{\sigma_1, \sigma_2 = \pm 1} \Big[ \id_{[0,T]}(\eta_1+\sigma_1|s(q_1,q_1')+2\pi l|) \, \id_{[0,T]}(\eta_2+\sigma_2 |s(q_2,q_2')+2\pi m|) \nonumber\\
& ~\times~ a(\eta_1+\sigma_1|s(q_1,q_1')+2\pi l|)  \, a(\eta_2+\sigma_2|s(q_2,q_2')+2\pi m|)\nonumber\\
& ~\times~\frac{1}{\sin(s(q_1',q_2'))\,} (f\times \chi)(\eta_1+\sigma_1|s(q_1,q_1')+2\pi l|,q_1',\eta_2+\sigma_2|s(q_2,q_2')+2\pi m|,q_2') \Big].
\label{eq:inteq3dclosedflrwgeneralk3}
\end{align}
Note that as an element of $\Banach$, $\chi$ is an equivalence class of functions modulo changes on sets of measure zero. In order to understand expressions such as $\chi(\eta_1+\sigma_1|s(q_1,q_1')+2\pi l|,q_1',\eta_2+\sigma_2|s(q_2,q_2')+2\pi m|,q_2')$ we choose an arbitrary representative of that class. Then we show that the integral operator in \eqref{eq:inteq3dclosedflrwgeneralk3} is bounded, with a bound that does not depend on the choice of the representative. In particular, this implies that \eqref{eq:inteq3dclosedflrwgeneralk3} is well-defined on $\Banach$.

The integral equation \eqref{eq:inteq3dclosedflrwgeneralk3} has the abstract structure
\be
	\chi = \chi^\free + \widehat{K} \chi,
\ee
where $\widehat{K} = \sum_{l, m, \sigma_1, \sigma_2} \widehat{K}_{\sigma_1\sigma_2}^{l m} $ and $\widehat{K}_{\sigma_1\sigma_2}^{l m} $ denotes the integral operator on the right hand side of \eqref{eq:inteq3dclosedflrwgeneralk3} that corresponds to the summand with fixed $l, m, \sigma_1,\sigma_2$.
We shall now show that each  $\widehat{K}_{\sigma_1\sigma_2}^{l m} $ is bounded and give an estimate for its operator norm. This estimate will be  independent of $\sigma_1, \sigma_2, l, m$.

Consider the norm of $\widehat{K}_{\sigma_1\sigma_2}^{l m} \chi$. Using the Cauchy-Schwarz inequality for the $d \Omega_3(q_1')\, d \Omega_3(q_2')$-integral as well as replacing $a$ and $f$ with their suprema implies:
\begin{align}
	& \| \widehat{K}_{\sigma_1\sigma_2}^{l m} \chi \|_\Banach \leq \esssup_{\eta_1,\eta_2 \in [0,T]} \frac{|\lambda|}{4(4\pi)^2} \, \| a\|^2_\infty \, \| f\|_\infty \left[ \int d \Omega_3(q_1)\, d \Omega_3(q_2) \right. \nonumber\\
&~ \times ~ \left( \int d \Omega_3(q_1')\, d \Omega_3(q_2')  \, \frac{1}{\sin^2 (s(q_1',q_2'))}  \right)\nonumber\\
&~ \times ~ \left( \int d \Omega_3(q_1')\, d \Omega_3(q_2')  \, \frac{1}{\sin^2 (s(q_1,q_1'))}  \, \frac{1}{\sin^2 (s(q_2,q_2'))}\right.\nonumber\\
&~ \times ~ \id_{[0,T]}(\eta_1+\sigma_1|s(q_1,q_1')+2\pi l|) \, \id_{[0,T]}(\eta_2+\sigma_2 |s(q_2,q_2')+2\pi m|)\nonumber\\
&~ \times ~|\chi|^2(\eta_1+\sigma_1|s(q_1,q_1')+2\pi l|,q_1',\eta_2+\sigma_2|s(q_2,q_2')+2\pi m|,q_2')  \Big)\Big]^{1/2}.
	\label{eq:closedcalc1}
\end{align}

We first consider the integral in the second line. Note that the geodesic distance of two points $q_1', q_2'$ on the 3-sphere is given simply by the angle $\alpha$ on the great circle that passes through $q_1', q_2'$. Hence, we can choose that angle $\alpha$ as one of the angles of the hyperspherical coordinates for the $d\Omega_3(q_2')$-integration for any fixed $q_1'$. That means, we use the coordinates
\begin{align}
	(q_2')^0 &= \cos \alpha_2,\nonumber\\
	(q_2')^1 &= \sin \alpha_2  \, \cos \beta_2, \nonumber\\
	(q_2')^2  &= \sin \alpha_2  \, \sin \beta_2 \,\sin \varphi_2, \nonumber\\
	(q_2')^3 &= \sin \alpha_2  \, \sin \beta_2 \, \sin \varphi_2,
\end{align}
where $\alpha_2, \beta_2 \in [0,\pi)$ and $\varphi_2 \in [0,2\pi)$. The surface element is given by $d\Omega_3(q_2') = \sin^2 \alpha_2 \, \sin \beta_2 \, d \alpha_2\, d \beta_2 \, d\varphi_2$.
 Then:
\begin{align}
	\int d \Omega_3(q_1')\, d \Omega_3(q_2')  \, \frac{1}{\sin^2 (s(q_1',q_2'))}  &= \int d \Omega_3(q_1') \int_0^\pi d \alpha_2 \int_0^\pi d \beta_2 \int_0^{2\pi} d \varphi_2 \, \frac{\sin^2 \alpha_2 \, \sin \beta_2}{\sin^2 \alpha_2}\nonumber\\
&= |S^3| \times 4\pi^2 = 8\pi^4.
\end{align}
Here, we have used $|S^3|  = 2\pi^2$. Hence, \eqref{eq:closedcalc1} becomes
\begin{align}
	& \| \widehat{K}_{\sigma_1\sigma_2}^{l m} \chi \|_\Banach \leq \esssup_{\eta_1,\eta_2 \in [0,T]} \frac{|\lambda|}{16\sqrt{2}} \, \| a\|^2_\infty \, \| f\|_\infty \left[\int d \Omega_3(q_1)\, d \Omega_3(q_2) \, d \Omega_3(q_1')\, d \Omega_3(q_2')   \right.\nonumber\\
&~ \times ~ \frac{1}{\sin^2 (s(q_1,q_1'))}  \, \frac{1}{\sin^2 (s(q_2,q_2'))}\nonumber\\
&~ \times ~ \id_{[0,T]}(\eta_1+\sigma_1|s(q_1,q_1')+2\pi l|) \, \id_{[0,T]}(\eta_2+\sigma_2 |s(q_2,q_2')+2\pi m|)\nonumber\\
&~ \times ~|\chi|^2(\eta_1+\sigma_1|s(q_1,q_1')+2\pi l|,q_1',\eta_2+\sigma_2|s(q_2,q_2')+2\pi m|,q_2') \Big]^{1/2}.
	\label{eq:closedcalc2b}
\end{align}
Now we exchange the order of the $d \Omega_3(q_1')\, d \Omega_3(q_2')$ and the $d \Omega_3(q_1)\, d \Omega_3(q_2)$-integrals. In the $d \Omega_3(q_1)\, d \Omega_3(q_2)$-integral we then change variables such that $s(q_i,q_i')$ is one of the angles of the hyperspherical coordinates. We obtain:
\begin{align}
	& \| \widehat{K}_{\sigma_1\sigma_2}^{l m} \chi \|_\Banach \leq \esssup_{\eta_1,\eta_2 \in [0,T]} \frac{|\lambda|}{16\sqrt{2}} \, \| a\|^2_\infty \, \| f\|_\infty \left[\int d \Omega_3(q_1')\, d \Omega_3(q_2')   \right.\nonumber\\
&~ \times ~\int_0^\pi d \alpha_1 \int_0^\pi d \beta_1 \int_0^{2\pi} d \varphi_1 \, \frac{\sin^2 \alpha_1 \, \sin \beta_1}{\sin^2 \alpha_1} \int_0^\pi d \alpha_2 \int_0^\pi d \beta_2 \int_0^{2\pi} d \varphi_2 \, \frac{\sin^2 \alpha_2 \, \sin \beta_2}{\sin^2 \alpha_2} \nonumber\\
&~ \times ~\id_{[0,T]}(\eta_1+\sigma_1|\alpha_1+2\pi l|) \, \id_{[0,T]}(\eta_2+\sigma_2 |\alpha_2+2\pi m|)\nonumber\\
&~ \times ~|\chi|^2(\eta_1+\sigma_1|\alpha_1+2\pi l|,q_1',\eta_2+\sigma_2|\alpha_2+2\pi m|,q_2') \Big]^{1/2}.
	\label{eq:closedcalc3b}
\end{align}
We now exchange the order of the integrations again and use the $d \Omega_3(q_1')\, d \Omega_3(q_2')$-integral to express the spatial norm of $\chi$. This yields:
\begin{align}
	& \| \widehat{K}_{\sigma_1\sigma_2}^{l m} \chi \|_\Banach \leq \esssup_{\eta_1,\eta_2 \in [0,T]} \frac{|\lambda|}{16\sqrt{2}}\, \| a\|^2_\infty \, \| f\|_\infty\nonumber\\
&~ \times ~\left[ \int_0^\pi d \alpha_1 \int_0^\pi d \beta_1 \int_0^{2\pi} d \varphi_1 \, \sin \beta_1 \int_0^\pi d \alpha_2 \int_0^\pi d \beta_2 \int_0^{2\pi} d \varphi_2 \, \sin \beta_2 \right.\nonumber\\
&~ \times ~\id_{[0,T]}(\eta_1+\sigma_1|\alpha_1+2\pi l|) \, \id_{[0,T]}(\eta_2+\sigma_2 |\alpha_2+2\pi m|)\nonumber\\
&~ \times ~\|\chi(\eta_1+\sigma_1|\alpha_1+2\pi l|,\cdot,\eta_2+\sigma_2|\alpha_2+2\pi m|,\cdot) \|^2_{L^2((S^3)^2)} \Big]^{1/2}.
	\label{eq:closedcalc4}
\end{align}
Next, we replace the expressions in last two lines by their essential suprema. This results in:
\begin{align}
	& \| \widehat{K}_{\sigma_1\sigma_2}^{l m} \chi \|_\Banach \leq \esssup_{\eta_1,\eta_2 \in [0,T]} \frac{|\lambda|}{16\sqrt{2}} \, \| a\|^2_\infty \, \| f\|_\infty \times 4\pi^2 \, \|\chi \|_\Banach \nonumber\\
&~ = \frac{\pi^2}{4\sqrt{2}} \, |\lambda| \, \| a\|^2_\infty \, \| f\|_\infty \,\|\chi \|_\Banach.
	\label{eq:closedcalc5}
\end{align}

This shows that each $\widehat{K}_{\sigma_1\sigma_2}^{l m}$ is bounded with the same norm. Now, there are four possible values of the pair $(\sigma_1, \sigma_2)$. Furthermore, as noted below \eqref{eq:closedgreensfn}, there are at most $\left(\left\lfloor \frac{T}{\pi}\right\rfloor +1\right)$ possible values of $l,m$ such that $\widehat{K}_{\sigma_1\sigma_2}^{l m}$ is non-zero. Overall, we obtain the estimate:
\be
	\| \widehat{K} \|_{\Banach \rightarrow \Banach} \leq  4 \left(\left\lfloor \frac{T}{\pi}\right\rfloor +1\right)^2 \, 	\| \widehat{K}_{11}^{00} \|_{\Banach \rightarrow \Banach} \frac{\pi^2}{\sqrt{2}} \, |\lambda| \,\left(\left\lfloor \frac{T}{\pi}\right\rfloor +1\right)^2 \, \| a\|^2_\infty \, \| f\|_\infty.
\ee
Thus, for
\be
	|\lambda| < \left( \frac{\pi^2}{\sqrt{2}}\,\left(\left\lfloor \frac{T}{\pi}\right\rfloor +1\right)^2 \, \| a\|^2_\infty \, \| f \|_\infty\right)^{-1}\,,
\ee
$\widehat{K}$ defines a contraction on $\Banach$ and Banach's fixed point theorem ensures the existence of a unique solution $\chi \in \Banach$ of \eqref{eq:inteq3dclosedflrwgeneralk3} for every $\chi^\free \in \Banach$. As \eqref{eq:inteq3dclosedflrwgeneralk3} is equivalent to \eqref{eq:inteq3dclosedflrwgeneralk}, the claim follows. \qed
\end{proof}

%%%
\paragraph{Remark.} As in the flat FLRW case, the theorem covers certain fully covariant interaction kernels (e.g., bounded interaction kernels that depend only on the time-like distance of $x_1,x_2$ and vanish when $x_1,x_2$ are not time-like related). Furthermore, it seems remarkable to us that it is possible to give rigorous meaning to a dynamics where interactions depend both on the past and the future (as they do in the classical Wheeler-Feynman theory \cite{WF1,WF2}). Note that contrary to the flat FLRW case with retarded Green's functions, $\psi^\free$ does not play the role of initial data on the Big Bang. Nevertheless, there is a one-to-one correspondence between free solutions $\psi^\free$ and solutions of the integral equation \eqref{eq:inteq3dclosedflrwgeneralk}. One can thus still understand the integral equation as determining a correction to $\psi^\free$ as a consequence of interaction.

%%%%%%%%%%%%
%%%%%%%%%%%%
\section{Conclusions} \label{sec:conclusion}

Here we have shown how the multi-time integral equation \eqref{eq:inteq}, which describes two directly interacting relativistic quantum particles, can be extended to curved spacetimes in a canonical way. Our motivation for doing this has been the fact that for spacetimes with a Big Bang singularity, the time integrals in \eqref{eq:inteq} do not extend to $-\infty$, and even become finite in the case of retarded interactions. In the case of time-symmetric interactions the time integrals become finite if, in addition, the spacetime has a Big Crunch singularity. Finite time integrals, in turn, make the integral equation easier to understand from a mathematical point of view and avoid potential divergences which could occur otherwise (see \cite[sec. 2.2]{mtve}).

The main difficulty with formulating \eqref{eq:inteq} on curved spacetimes is to obtain an explicit expression for the Green's functions that occur in the equation. It has been demonstrated how to obtain such explicit expressions for the massless Klein-Gordon equation with coupling to the scalar curvature on general FLRW spacetimes. That equation is conformally invariant, which has made it possible to calculate the Green's functions in a straightforward way. This has led us to the integral equations \eqref{eq:inteq1dflatflrw}, \eqref{eq:inteq2dflatflrw}, \eqref{eq:inteq3dflatflrw} on flat FLRW spacetimes with 1+1, 1+2, and 1+3 dimensions, as well as \eqref{eq:inteq3dopenflrw} on open and \eqref{eq:inteq3dclosedflrw} on closed FLRW spacetimes with 1+3 dimensions.

Our main results are existence and uniqueness theorems for simplified versions of these integral equations (simplified in the sense that interaction does not happen exactly along light cones but on extended spacetime regions). The results cover flat FLRW spacetimes and purely retarded interactions, as well as closed FLRW spacetimes and time-symmetric interactions. Open FLRW spacetimes have not been treated here as they are expected to be qualitatively similar to flat FLRW spacetimes. They could be covered by re-doing the proof in \cite{mtve} while replacing the flat 3-space $\R^3$ with the hyperbolic space $\mathbb{H}^3$.
Our results also include a characterization of the behavior of $\psi$ towards the spacetime singularities, as well as a parametrization of the solution space. This parametrization works via solutions of the free equations. This suggests to read the integral equation \eqref{eq:inteq} as determining a correction to a free solution as the consequence of interaction. In case of retarded interactions, we have seen that this way of classifying the solutions is equivalent to Cauchy data at the Big Bang singularity.

One may ask why we have taken the effort of doing explicit calculations with the Green's functions instead of viewing the integral equation \eqref{eq:inteq} as an abstract operator equation of the form
\be
	\big( 1 - \lambda\, (\widehat{G}^\ret \otimes \widehat{G}^\ret) \widehat{K}\big) \psi = \psi^\free
	\label{eq:abstractinteq}
\ee
where $\widehat{G}^\ret$ is the operator that convolves with the retarded Green's function $G^\ret$ and $\widehat{K}$ the multiplication operator with $K$. It is then easy to see that $\widehat{G}^\ret$ defines a bounded operator on $\Banach_d$, and consequently, if $\widehat{K}$ is bounded, the Neumann series yields the unique solution of \eqref{eq:abstractinteq} for
\be
	|\lambda| < \| \widehat{G}^\ret \|^{-2} \, \| \widehat{K} \|^{-1}.
\ee
That means, one obtains a general existence and uniqueness result with little effort. However, the result is rather weak. Our theorems for the case of retarded Green's functions hold for arbitrary $\lambda$. Furthermore, it is possible to treat unbounded multiplication operators $\widehat{K}$. To obtain a result for arbitrary $\lambda$ is especially important as the bound for $\| \widehat{G}^\ret \|$ may depend on $T$ (the time up to which one solves the equation), at least for the physically natural function space $\Banach_d$. Then a smallness assumption for $\lambda$ amounts to a short-time existence result. As we have shown, one can do much better than this by using the Volterra structure of the integral equation in the retarded case. Even in the time-symmetric case, where we do assume a smallness condition for $\lambda$, we are nevertheless able to treat certain unbounded operators $\widehat{K}$.

An existence and uniqueness result without a smallness condition on $\lambda$ could also be obtained if one could show that $(\widehat{G} \otimes \widehat{G}) \widehat{K}$ is a compact operator on $\Banach_d$. As the spectrum of a compact operator is discrete, this would yield the existence and uniqueness of \eqref{eq:abstractinteq} for almost all values of $\lambda$. If one could in addition prove uniqueness, then the Fredholm alternative would also imply the existence of solutions. While this strategy sounds promising in principle (especially to improve the result in the time-symmetric case), we do not, at present, know a good method to decide when $(\widehat{G}\otimes \widehat{G}) \widehat{K}$ is compact.

Despite being simplified, the class of equations for which our results are valid does include manifestly covariant possibilities. This is remarkable, considering that it is a difficult and long-standing problem in mathematical physics to rigorously construct interacting relativistic quantum dynamics in 1+3 spacetime dimensions. Our work opens up a possible new approach to solving this problem. We emphasize that the multi-time wave function is a crucial resource for this approach; without it, one could not have expressed direct relativistic interactions with time delay.

In the future, it would be desirable to extend our work in the following directions:
\begin{itemize}
	\item To use the Dirac equation instead of the Klein-Gordon equation as the free wave equation (see \cite{lienert_noeth_2019}),
   \item To consider the physically most natural (and highly singular) interaction kernels $K(x_1,x_2) = G^\sym(x_1,x_2)$ for 1+2 and especially 1+3 spacetime dimensions, and
	\item To extend the existence and uniqueness results to $N$ particles. (The $N$ particle generalization of the integral equation \eqref{eq:inteq} has been discussed in \cite{direct_interaction_quantum}.)
\end{itemize}

%%%
\paragraph{Acknowledgments.}~\\[1mm]
We would like to thank Fioralba Cakoni for helpful discussions.\\[1mm]
\begin{minipage}{15mm}
\includegraphics[width=13mm]{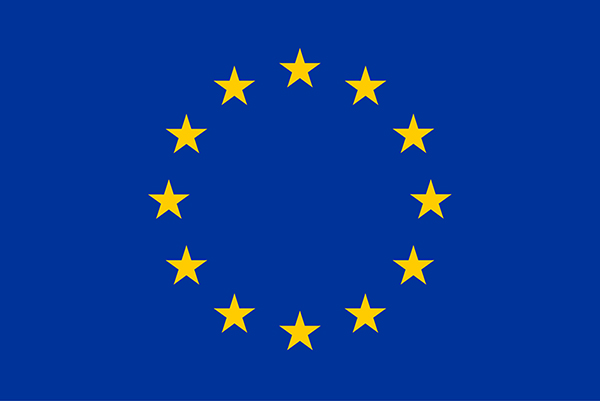}
\end{minipage}
\begin{minipage}{143mm}
This project has received funding from the European Union's Framework for Re-\\
search and Innovation Horizon 2020 (2014--2020) under the Marie Sk{\l}odowska-
\end{minipage}\\[1mm]
Curie Grant Agreement No.~705295.

%\bibliography{literature}

\end{document}